\documentclass[a4paper,onecolumn,notitlepage,nofootinbib,longbibliography,superscriptaddress, floatfix, pra, aps, 10pt]{revtex4-2}

\usepackage{comment}
\usepackage[caption=false]{subfig}
\usepackage[normalem]{ulem}
\usepackage{
    adjustbox,
    algorithm,
    algpseudocode,
    amscd,
    amsfonts,
    amsmath,
    amssymb,
    amsthm,
    bm,
    booktabs,
    csquotes,
    enumerate,
    enumitem,
    graphicx,
    IEEEtrantools,
    indentfirst,
    latexsym,
    mathrsfs,
    mathtools,
    multirow,
    nicefrac,
    placeins,
    pifont,
    relsize,
    tabls,
    tikz-cd,
    times,
    comment
}

\usepackage{dsfont}
\usepackage{physics}

\newtheorem{lemma}{Lemma}
\newtheorem{theorem}{Theorem}

\newtheorem{definition}{Definition}

\newcommand{\aqa}{$\langle aQa^L \rangle$ Applied Quantum Algorithms, Leiden University, The Netherlands}
\newcommand{\liacs}{LIACS, Leiden University, Niels Bohrweg 1, 2333 CA, Leiden, The Netherlands}
\newcommand{\bmw}{BMW Group, 80788 München, Germany}

\begin{document}

\title{Note on the Universality of Parameterized IQP Circuits with Hidden Units for Generating Probability Distributions}
\date{\today}

\author{Andrii Kurkin}
\affiliation{\aqa}%
\affiliation{\liacs}%
\affiliation{\bmw}%

\author{Kevin Shen}
\affiliation{\aqa}%
\affiliation{\liacs}%
\affiliation{\bmw}%

\author{Susanne Pielawa}
\affiliation{\bmw}%

\author{Hao Wang}
\affiliation{\aqa}%
\affiliation{\liacs}%

\author{Vedran Dunjko}
\affiliation{\aqa}%
\affiliation{\liacs}%
 
\begin{abstract}
    In a series of recent works, an interesting quantum generative model based on parameterized instantaneous polynomial quantum (IQP) circuits has emerged as they can be trained efficiently classically using any loss function that depends only on the expectation values of observables of the model. The model is proven not to be universal for generating arbitrary distributions, but it is suspected that marginals can be - much like Boltzmann machines achieve universality by utilizing hidden (traced-out in quantum jargon) layers. In this short note, we provide two simple proofs of this fact. The first is near-trivial and asymptotic, and the second shows universality can be achieved with a reasonable number of additional qubits.
\end{abstract}

\maketitle

\section{Introduction}\label{s:intro}
Recently, a promising quantum generative model has been proposed and demonstrated to be efficiently trainable at scale \cite{recio2025train, recioarmengol2025iqpoptfastoptimizationinstantaneous}. The model is a type of Quantum Circuit Born Machine (QCBM) \cite{Liu_2018, Benedetti_2019}, with the circuit specifically chosen as a parameterized instance of Instantaneous Quantum Polynomial (IQP) circuit \cite{Shepherd_2009, Bremner_2010}. The model prepares an $n$-qubit state from a parameterized IQP circuit, and measuring each qubit in the computational basis will generate a bitstring of length $n$. We will call this model IQP-QCBM.

In Ref.~\cite{recio2025train}, it was shown that an IQP-QCBM with $n$ qubits cannot represent arbitrary probability distributions over $n$-bitstrings. Although the model provides \( 2^n -1 \) complex phases, these are insufficient for universality. However, the authors pointed out a potential pathway toward universality:
the model may be universal with additional qubits, i.e., by applying circuits on $m+n$ qubits and considering the
marginal distributions on $n$ qubits (tracing out the first $m$ qubits)~\cite{recio2025train}.

In this short note, we provide two simple constructive proofs that IQP-QCBMs with additional qubits are indeed universal. 
In Section \ref{s:prelims}, we give definitions needed for the proof. In Section \ref{s:results}, we present the main results. We close with a discussion in Section \ref{s:discussion}.

\section{Preliminaries}\label{s:prelims}

\begin{definition}[Parameterized IQP-circuit \cite{recio2025train, recioarmengol2025iqpoptfastoptimizationinstantaneous}]
A parameterized IQP circuit on $n$ qubits is a circuit with the following construction:
\begin{enumerate}
    \renewcommand{\labelenumi}{(\roman{enumi})}
    \item initialize in the state $\lvert0^{\otimes n} \rangle$,
    \item apply a sequence of parameterized gates of the form $\exp(i\theta_j X_j)$, where $X_j$ is a tensor product of Pauli $X$ operators acting on a subset of the qubits, 
    \item measure all qubits in the computational basis.
    \end{enumerate}
\end{definition}

The model above is equivalent to a model where, in place of $(ii)$, we consider a three-layer model where we first apply a layer of Hadamard gates, followed by arbitrary diagonal operators (with respect to the computational basis), and a final layer of Hadamard operators.
In our derivations, we will be utilizing this representation. 

\begin{definition}[Parameterized IQP circuit with hidden units]
A parameterized IQP circuit with hidden units is a parameterized IQP circuit in which a chosen subset of qubits is traced out.
\end{definition}
We have here adopted the terminology from the model of Boltzmann machines where analogously using additional neurons that are traced out (marginalized), we obtain universality (see \cite{6796877}).

\begin{definition}[$k$-sparse distribution]
Consider a probability distribution $p$ on the discrete sample space $\mathcal{X}$. $p$ is said to be $k$-sparse if $p$ has nonzero probability mass for at most $k$ outcomes, i.e.,
\begin{equation}
\lvert\{x\in \mathcal{X}|p(x)>0\}\rvert \leq k
\end{equation}
\end{definition}

\section{Results}\label{s:results}
\subsection{Approximate universality}
We first provide a very simple proof of asymptotic universality for IQP-QCBMs with hidden units, where we only require real-valued ($\theta \in \{0, \pi\}$) phases.
\begin{lemma}[Asymptotic approximate universality of IQP-QCBM with hidden units]
For any probability distribution $p$ over $\{0,1\}^n$, there exists a parameter setting $\theta = \{\theta_{j,k}\}_{j \in \{0,1\}^m, k \in \{0,1\}^n}$ of the IQP circuit acting on $m + n$ qubits  such that the marginal distribution $q_{\theta}$ on the $n$ final qubits satisfies:
\begin{equation*}
    \delta(p, q_{\theta}) \in O\left(\frac{1}{2^{m-n}}\right)
\end{equation*}
where $\delta$ denotes the total variation distance between distributions.
\end{lemma}

\begin{proof}
Consider the quantum state before the last layer of Hadamard gates. It can be written as:
\begin{equation}
    \ket{\psi} =  \frac{1}{\sqrt{2^{m+n}}}\sum_{j \in \{0,1\}^m} \sum_{k \in \{0,1\}^n} e^{i \theta_{j,k}} |j\rangle |k\rangle,
\end{equation}
where the phases $\theta_{j,k}$ are free. Now consider the state:
\begin{equation}
    \ket{\psi'} = \frac{1}{\sqrt{2^{m+n}}} \sum_{j = 0}^{2^m-1} \sum_{k=0}^{2^n-1} |j\rangle |v(k)\rangle,
\end{equation}
where $v(k)$ is an arbitrary function from $\{0, \dots, 2^m-1\}$ to $\{0,1\}^n$.

The probability of measuring some bitstring $b$ on the last $n$ qubits of the state $\ket{\psi'}$ is given by:
\begin{equation}
    P(b) = \frac{| \{ j \in \{0,1\}^m \mid v(j) = b \}|}{2^m}.
\end{equation}
In the limit of large $m$, we can choose the function $v(j)$ such that the probability distribution $P(b)$ approaches the desired probability distribution to arbitrary precision with respect to the infinity norm (and thus also total variation distance).

Finally, we notice that if we apply Hadamard gates to each of the last $n$ qubits of $\ket{\psi'}$, we obtain a state of the form:
\begin{equation}
    \frac{1}{\sqrt{2^{m+n}}} \sum_{j = 0}^{2^m-1} \sum_{k=0}^{2^n-1} |j\rangle |h(j)\rangle,
\end{equation}
where each $|h(j)\rangle$ is a string of Pauli-X eigenstates, so a superposition of all computational basis states with differing signs. 

Thus, the overall state is in the form of $\ket{\psi}$, making it a valid state of an IQP-QCBM (which is an UMA state; see Lemma 3). The final layer of Hadamards on all qubits ensures the desired marginal distribution.

We now derive the approximation error in total variation distance as a function of $m$. Our model generates a restricted class of distributions where each component has $m$-bit precision. Nevertheless, we prove that any distribution can still be approximated by our model with total variation error $O\left(\frac{1}{2^{m-n}}\right)$

Note, that any probability $p_j$ of a probability distribution over $n$-bit strings can be expressed in the form $p_j = \frac{i_j}{2^m} + \varepsilon_j$, where $i_j \in \mathbb{N}, (2^m \times \varepsilon_j) \in [0, 1]$, $\sum_j \frac{i_j}{2^m} + \varepsilon_j = 1$, and $r := 2^m \sum_j \varepsilon_j \in \mathbb{N}$. We present an algorithm that constructs a distribution $q$ which can be output using our construction above using $m$ hidden units, such that $\delta(p, q) \in O\left(\frac{1}{2^{m-n}}\right)$.

\begin{equation}
    q_j = \begin{cases}
           \frac{i_j}{2^m} + \frac{1}{2^m} & \text{if } j \leq r\\
           \frac{i_j}{2^m} & \text{if } j > r 
        \end{cases}
\end{equation}

It is easy to verify that $q$ is a valid probability distribution since $\sum_j q_j = 1$.

\begin{equation}
    \delta(p, q) = \frac{1}{2} \sum_{j=1}^{2^n} |p_j - q_j| 
    = \frac{1}{2} \left( \sum_{j=1}^r \left| \frac{1-\varepsilon_j}{2^m} \right| 
    + \sum_{j=r+1}^{2^n} \frac{\varepsilon_j}{2^m} \right) 
    \leq \frac{1}{2} \left(\frac{r}{2^m} + \frac{2^n - r}{2^m}\right) 
    = \frac{1}{2} \frac{1}{2^{m-n}}
\end{equation}

\end{proof}
Of course, this statement is not the most efficient possible, as the construction utilizes only $0$ and $\pi$ relative phases. By using all possible phases, a stronger result can be obtained as we show next.

\subsection{Exact universality}

\begin{lemma}\label{lemma: rho_b decomposition}
    Let the total quantum system consist of $m+n$ qubits, where the quantum state is prepared using an IQP circuit, specified by diagonal matrix $D(\theta) = \sum_{x \in \{0,1\}^m} \sum_{y \in \{0,1\}^n} e^{i\theta_{x, y}} |x\rangle \langle x| \otimes |y\rangle \langle y|$ with the final layer of Hadamard gates omitted. Then, the reduced density matrix $\rho_2$ of the second, $n$-qubit, subsystem is given by
    \begin{equation}
        \rho_2 = \frac{1}{2^m} \sum_{k \in \{0,1\}^m} |\psi_k\rangle \langle \psi_k|,
    \end{equation}
    where each state $|\psi_k\rangle$ takes the form
    \begin{equation}
        |\psi_k\rangle = \frac{1}{\sqrt{2^n}} \sum_{y \in \{0,1\}^n} e^{i\theta_{k, y}} |y\rangle,\label{eq: psi}
    \end{equation}
    and the phases $\theta_{k, y}$ are trainable parameters of the IQP circuit.
\end{lemma}

\begin{proof}
    \begin{equation}
    \rho_{1+2} = D (\theta)H^{\otimes (m+n)}|0\rangle^{\otimes (m+n)} \langle 0|^{\otimes (m+n)} H^{\otimes (m+n)} D^{\dagger} (\theta)= \frac{1}{2^{m+n}} \sum_{x,x'\in \{0,1\}^n} \sum_{y,y' \in \{0,1\}^m} e^{i(\theta_{x,y} - \theta_{x',y'})} |x\rangle |y\rangle \langle x'| \langle y'|
\end{equation}

Reduced density matrix for the second subsystem:
\begin{align}
    \rho_2 &= \sum_{k \in \{0,1\}^n} \frac{1}{2^{m+n}} \sum_{x,x'\in \{0,1\}^m} \sum_{y,y' \in \{0,1\}^n} e^{i(\theta_{x,y} - \theta_{x',y'})} \langle k|x\rangle \langle x'|k\rangle \otimes |y\rangle \langle y'| = \notag \\
    &= \sum_{y,y'\in \{0,1\}^n} \left(\frac{1}{2^{m+n}} \sum_{k \in \{0,1\}^m} e^{i(\theta_{k,y} - \theta_{k,y'})} \right) |y\rangle \langle y'| = \notag \\&= \frac{1}{2^m}\sum_{k \in \{0,1\}^m} \underbrace{\left(\frac{1}{\sqrt{2^n}}\sum_{y\in \{0,1\}^n} e^{i\theta_{k,y}}|y\rangle \right)}_{|\psi_k\rangle} \left(\frac{1}{\sqrt{2^n}}\sum_{y'\in \{0,1\}^n} e^{-i\theta_{k,y'}}\langle y'| \right) = \frac{1}{2^m} \sum_{k \in \{0,1\}^m} |\psi_k\rangle \langle \psi_k|
\end{align}
\end{proof}

In the below we use tilde to denote bitstrings specifying states which are strings of Pauli-X eigenstates.

\begin{lemma}\label{lemma: 2-sparse distributin with psi_k}
For any quantum state $|\psi_k\rangle = \alpha |\tilde{s}_1\rangle + \beta |\tilde{s}_2\rangle$ where $\tilde{s}_1$, $\tilde{s}_2 \in \{+,- \}^n$, there exist parameter setting $\theta$ in the state defined in Eq.~\eqref{eq: psi}, such that
$\bra{\tilde{s}_1}\ket{\psi_k} =  |\alpha|^2$.
\end{lemma}
\begin{proof}
We prove it in two steps. We say a normalized quantum state is a \emph{UMA} (uniform magnitude) state if it has all amplitudes of the same absolute value (with respect to the computational basis).  
Note that all states which are strings of $X$ eigenstates are UMA.\\
Next, note that for all probabilities $p \in [0,1]$, there exist $\alpha$, $\beta$ such that  
$\alpha |+\rangle + \beta |-\rangle$ is UMA.  
To see this, define
\[
|+_{\theta}\rangle = \frac{1}{\sqrt{2}} (|0\rangle + e^{i\theta} |1\rangle)
\]
and compute
\[
|\langle + | +_{\theta} \rangle|^2 = \cos^2(\theta/2).
\]
The right-hand side can attain all possible $p$ values. \\
Now we claim: if $\alpha$ and $\beta$ are such as above, then for all strings $\ket{\tilde{s}_1}$ and $\ket{\tilde{s}_2}$ of $X$ eigenstates, the state 
\[
|\psi\rangle = \alpha \ket{\tilde{s}_1} + \beta\ket{\tilde{s}_2}
\]
is UMA.

If $\tilde{s}_1 = \tilde{s}_2$, then the claim is trivial as we can set $p = 1$.

If not, let $j$ be the first position where $(\tilde{s}_1)_j \neq (\tilde{s}_2)_j$. We can assume $j = 1$. If not, we can swap the first and $j$-th qubit - this state will be UMA if and only if $|\psi\rangle$ was UMA.

So without loss of generality, we can write:  
\[
|\psi\rangle = \alpha |+\rangle |\tilde{s}_1'\rangle + \beta |-\rangle |\tilde{s}_2'\rangle.
\]
and it will suffice to prove this is an UMA state.
If $\tilde{s}_1' = \tilde{s}_2'$, we are done as $|\psi\rangle$ factorizes, and the tensor product of a UMA state (on the first qubit by previous analysis) and another UMA state is UMA.

If $\tilde{s}_1' \neq \tilde{s}_2'$, define $U$ as the $(n-1)$-qubit diagonal unitary such that  
\[
|\tilde{s}_2'\rangle = U |\tilde{s}_1'\rangle.
\]
Note that $U$ is a diagonal unitary with $\pm1$ on the diagonal.  
Define the Hadamard-controlled-$U$:
\[
\text{hc}U = |+\rangle\langle+| \otimes I_{n-1} + |-\rangle\langle-| \otimes U.
\]

Note that:  
(1)  
\[
\text{hc}U (\alpha |+\rangle |\tilde{s}_1'\rangle + \beta |-\rangle |\tilde{s}_1'\rangle) = |\psi\rangle,
\]  
and  
(2) hc$U$ is a block matrix with blocks $(I+U)/2$ on the diagonal and $(I-U)/2$ on the antidiagonal.

Note that in each row, we have only one non-zero element. So this is a permutation matrix. For a permutation matrix $U$, it holds that $|\psi\rangle$ is UMA if and only if $U|\psi\rangle$ is UMA.

So since  
\[
\alpha |+\rangle |\tilde{s}_1'\rangle + \beta |-\rangle |\tilde{s}_1'\rangle
\]
is UMA by previous analysis, so is $|\psi\rangle$.
\end{proof}

Next, we present the two key probability distribution decomposition lemmas in the language of probability vectors, i.e., vectors $p=(p_j)_j$, $\sum_j p_j = 1,\ p_j\geq 0$, encoding the probability of the measurement of each of the bitstrings (indexed by $j$).  
We will say a (probability) vector is $k$-sparse if it has at most $k$ non-zero entries.

\begin{lemma}
Every $N$-dimensional probability vector $p$ can be expressed as a uniform mixture of $N$ 3-sparse probability vectors. That is, there exists a set $\{ q(i)\}_i$ of $N$ 3-sparse probability vectors such that
$$
p = \sum_i \frac{1}{N} q(i).
$$
\end{lemma}

\textit{Proof.} We give a constructive proof by showing there exists an allocation matrix $Q\in\mathbb{R}^{N\times N}$ of $p$ such that the column sum $\sum_i Q_{ij} = p_j$ and the row sum $\sum_j Q_{ij} = 1/N$, which is also 3-row-sparse. The 3-sparse probability vectors are the rows of the $Q$, i.e., $q(i)/N$ is the $i$th row of $Q$.
We assume $p_1\leq p_2 \leq \cdots \leq p_N$ w.l.o.g. Let $k$ be the smallest index such that $p_{k+1} \geq 1/N$. Note $k=1$ if the distribution is uniform and can be as large as $N-1$, but not $N$. \\
The general idea is that we first allocate $p_1,\ldots, p_k$ to the first $k$ diagonal entries of $Q$. Since $p_1\leq p_2\leq\cdots \leq p_k \leq 1/N$, we have some capacity left on each row, and hence we can further try to spread out $p_{k+1}$ to the first $k$ rows. After $p_{k+1}$ is fully spread out, we proceed with $p_{k+2}$ and so on. Note that whenever the first rows are out of capacity, we will proceed using the remaining $N-k$ rows. The matrix $Q$ can be constructed and verified with the following three steps. \\
Step 1: Initialization.
$Q = 
\begin{bmatrix}
Q_0 &0 \\
0 & 0
\end{bmatrix}
,\;
Q_0 = 
\begin{bmatrix}
p_1 & 0   & \cdots & 0 \\
0   & p_2 & \cdots & 0 \\
\vdots & \vdots & \ddots & \vdots \\
0   & 0   & \cdots & p_k
\end{bmatrix}
$.
\\
Step 2: Iterate over $\ell=k+1, \ldots, N$. For each iteration, we spread out $p_\ell$ to the $\ell$ column of $Q$ as follows. For $i$ running from 1 to $N$, we set
\begin{equation}\label{eq:Q_update}
    Q_{i\ell} = \min\bigg\{\frac{1}{N} - \sum_{j}Q_{ij}, p_\ell - \sum_{\alpha} Q_{\alpha \ell}\bigg\}.
\end{equation}
Note that, $N^{-1} - \sum_{j}Q_{ij}$ is the remaining capacity of row $i$ (recall each row sums to $1/N$) and $p_\ell - \sum_{\alpha < i} Q_{\alpha \ell}$ is the residual of $p_\ell$ after spreading it over the first $i-1$ row. Two indicator functions check if a row is out of capacity or $p_\ell$ is completely spread out. \\
Step 3: Correctness. After iteration $\ell = k+1$, we must have at most two nonzero entries in each row since initially, $Q$ has at most one nonzero value per each row, and by Eq.~\eqref{eq:Q_update}, we only modify one entry of each row. Let $r$ be the largest row index such that $Q_{r\ell}\neq 0, \ell = k+1$. After the next iteration $\ell = k+2$, we have two cases: (1) if the capacity of row $r$ is used up by $p_{k+1}$, then the first nonzero entry of column $k+2$ starts at row $r$. In this case, $Q$ has at most two nonzero entries per row; (2) if some capacity of row $r$ is left, i.e., $\sum_{j} Q_{rj} \leq 1/N$, then first nonzero entry of column $k+2$ starts at row $r+1$, giving rise to three nonzero entries for this row. Note that the above argument between iterations $k+1$ and $k+2$ holds for any two iterations $\ell$ and $\ell+1$. Hence, we conclude that after iteration $\ell \geq k+1$, the following conditions are true:
\begin{enumerate}
    \item $\sum_{i} Q_{ij} = p_j$ for $1 \leq j \leq \ell$;
    \item There are at most three nonzero entries in each row of $Q$.
\end{enumerate} 
Now, proceed with the above argument until $\ell = N$, we must have: (1) $\sum_{j} Q_{ij} = 1/N$ for $1\leq i \leq N$; (2) $\sum_{i} Q_{ij} = p_j$ for $1 \leq j \leq N$; (3) there are at most three nonzero entries in each row of $Q$.
\qed

Next, we show that using twice as many $q$-vectors, we can represent the probability vector $p$ with 2-sparse probabilities. The basic idea is that a 3-sparse vector can always be written as a uniform mixture of two 2-sparse vectors: Let $p$ be a 3-sparse vector with $0 < p_a \leq p_b \leq p_c$, for indices $a, b, c$. Then we define the probability vectors $q_1$ and $q_2$ as follows: $(q_1)_a = 2p_a$, $(q_1)_c = 1 - 2 p_a$, $(q_2)_b = 2 p_b$, and $(q_2)_c = 1 - 2p_b$. The correctness can be verified by $q_1,q_2$ are 2-sparse, $q_1/2 + q_2/2 = p$, and $p_a\leq 1/2$ and $p_b \leq 1/2$.

\begin{lemma}\label{lemma: any distributin with 2-sparse}
Every $N$-dimensional probability vector $p$ can be expressed as a uniform mixture of~$2N$ 2-sparse probability vectors. That is, there exists a set $\{ q(i) \}_i$ of $2N$ 2-sparse probability vectors such that
$$
p = \sum_i \frac{1}{2N} q(i).
$$
\end{lemma}
\textit{Proof.} Taking Lemma 5, we know there exists a set $\{ q(i)\}_i$ of $N$ 3-sparse probability vectors such that $p = \sum_i \frac{1}{N} q(i)$. Now, for each $q(i)$, if it is 3-sparse but not 2-sparse, we can always split it into two 2-sparse vectors using the observation before the lemma, i.e., $q(i) = q_1(i)/2 + q_2(i)/2$; If $q(i)$ is  2-sparse, we simply produce two copies of it.
This construction satisfies the conditions of the lemma.  
\qed

With these lemmas in place, we can now prove the main theorem.

\begin{theorem}[Exact universality of IQP-circuits with hidden qubits]
For any target probability distribution $p$ over $\{0,1\}^n$, there exists an IQP circuit acting on $n$ visible qubits and $m = n+1$ hidden qubits that produces $p$ exactly when we measure the final $n$ qubits.
\end{theorem}
\begin{proof}
By Lemma~\ref{lemma: rho_b decomposition} the reduced density matrix of an IQP circuit with $m$ hidden qubits can be exactly written as sum of $2^m$ pure states density matrices $|\psi_k\rangle\langle\psi_k|$.
The output distribution will be a uniform mixture of the distributions obtained from measuring one of the $|\psi_k\rangle$ states in the basis specified by Pauli-X eigenstates. 
Each $|\psi_k\rangle$ can be expressed as a superposition of strings of Pauli-X eigenstates, and the probability of observing the string $b \in \{0,1\}^n$ is given by $|\langle \tilde{b}| \psi_k\rangle|^2$
with $|\tilde{b}\rangle = H^{\otimes n} |b\rangle  $.
In Lemma~\ref{lemma: 2-sparse distributin with psi_k} we have shown how we can encode any 2-sparse distribution in  $|\psi_k\rangle$.
 On the other hand, Lemma~\ref{lemma: any distributin with 2-sparse} states that any probability distribution over $\{0,1\}^n$ can be decomposed as a uniform sum of $2^{n+1}$ such 2-sparse distributions. Therefore, by choosing $m = n+1$, we have sufficient expressive power to represent any target distribution $p$ exactly using an IQP circuit with hidden qubits.
\end{proof}

\subsection{Remark on optimality}
We note that it is not clear if the construction of the main Theorem above is optimal in terms of the number of ancillary qubits, i.e., hidden units for the worst-case distributions. 
On the one hand, the purpose of this construction is to generate a particular reduced density matrix over the last $n$ qubits, and it is well known that this can be purified using at most $n$ additional qubits. However, in general, the states on the purifying system will then not correspond to UMA states, this purified state may not be obtainable by an IQP circuit in general.
An $n$ instead of $n+1$ qubit overhead could also be obtained if any $N$-dimensional probability vector could be decomposed to a uniform mixture of the sum of just $N$ 2-sparse probability vectors (we achieve 3-sparse). We do not have a strong argument that this is impossible, but we have equally failed to provide a provably correct construction.
Another route to reducing resources is to encode more than 2 outcome distributions in the states $|\psi_k\rangle$. However, already for 2 qubit (4 dimensional) case, the non-universality results of \cite{recio2025train} (page 47) imply that there exist distributions that cannot be encoded into UMA states. 
In general, any $n$-qubit UMA state characterized by $2^n-1$ phases $(\theta_j)_j$ (counting for the global phase) can only generate probability vectors with probabilities expressible with $p_k \propto | 1+ \sum_{j=2}^{2^n} (-1)^{z\cdot j} e^{i\theta_j} |^2 $, which, as said, do not include all distributions already for $n=2$. Already, one hidden unit allows us to recover arbitrary convex combinations of any two such distributions, and further hidden units allow for more complicated combinations. The minimal point at which this becomes universal we leave for future work.

\section{Discussion}\label{s:discussion}
This short note provides two simple proofs of the universality of IQP-QCBM when hidden (ancillary) qubits are used. 
This provides additional motivation to study these models as generative modes in near-term quantum computing. 
We, however, highlight that the purpose of this construction is theoretical, and should not be used as a suggestion to use large numbers of hidden units in practice. Since we deal with exponential outcome spaces, universality requires dealing with exponentially many trainable parameters, and in our construction, increasing the register size by a factor of 3. The latter in general would cause a quicker decay of gradients, and the prior makes the setting immediately intractable. Furthermore, fully universal models are often a bad idea as they likely lead to significant overfitting\footnote{What overfitting means in generative modelling is a subtle issue which we do not discuss here.}.
This work does further clarify the role and potential of using more complicated diagonal encoding operations, and additional hidden qubits as structural parameters which can be used to balance expressivity, overfitting, and trainability.

\section{Acknowledgements}

VD, HW, and AK acknowledge the support from the Dutch National Growth Fund (NGF), as part of the Quantum Delta NL programme. VD acknowledges support from the Dutch Research Council (NWO/OCW), as part of the Quantum Software Consortium programme (project number 024.003.03). This project was also co-funded by the European Union (ERC CoG, BeMAIQuantum, 101124342).

\bibliography{ref.bib}

\end{document}